\documentclass[11pt]{article}%
\usepackage{amsfonts}
\usepackage{amsmath}
\usepackage{amssymb}
\usepackage{graphicx}%
\providecommand{\U}[1]{\protect\rule{.1in}{.1in}}
\newtheorem{theorem}{Theorem}
\newtheorem{acknowledgement}[theorem]{Acknowledgement}

\newtheorem{corollary}[theorem]{Corollary}

\newenvironment{proof}[1][Proof]{\noindent\textbf{#1.} }{\ \rule{0.5em}{0.5em}}
\begin{document}

\title{Polar Duality and the Donoho--Stark Uncertainty Principle}
\author{Maurice de Gosson\\Austrian Academy of Sciences\\Acoustics Research Institute\\1010, Vienna, AUSTRIA\\and\\University of Vienna\\Faculty of Mathematics (NuHAG)\\1090 Vienna, AUSTRIA}
\maketitle

\begin{abstract}
Polar duality is a fundamental geometric concept that can be interpreted as a
form of Fourier transform between convex sets. Meanwhile, the Donoho--Stark
uncertainty principle in harmonic analysis provides a framework for comparing
the relative concentrations of a function and its Fourier transform. Combining
the Blaschke--Santal\'{o} inequality from convex geometry with the
Donoho--Stark principle, we establish estimates for the trade-off \ of
concentration between a square integrable function in a symmetric convex body
and that of its Fourier transform in the polar dual of that body. In passing,
we use the Donoho--Stark uncertainty principle to establish a new
concentration result for the Wigner function.

\end{abstract}

\textbf{Keywords}: quantum indeterminacy; polar duality; Blaschke--Santal\'{o}
inequality; Donoho--Stark uncertainty principle; quantum blobs; symplectic group

\section{Introduction}

The term \textquotedblleft quantum indeterminacy\textquotedblright\ refers to
a fundamental concept in quantum mechanics that highlights the intrinsic
uncertainty and unpredictably in the behavior of quantum systems. There are
several different ways to express quantum indeterminacy; the simplest (from
which actually most others are derived) is that a function and its Fourier
transform cannot simultaneously sharply located. On a more sophisticated
level, it is expressed by the Heisenberg uncertainty relations (or their
refinement, the Robertson--Schr\"{o}dinger inequalities). The drawback of the
latter is that they privilege variances (and covariances) for measuring
deviations;, this drawback has been discussed and criticized by Hilgevoord and
Uffink \cite{hi02,hiuf} who point out that their use for measuring the
deviations is only optimal for states that are Gaussian or nearly Gaussian
states. In previous work \cite{gopol1,gopol2,gopol3,gopol4,gopolar} we have
proposed a version of quantum indeterminacy using the geometric concept of
\emph{polar duality}. Polar duality is a concept from convex geometry, which
can be viewed (loosely) as a kind of Fourier transform between sets: if $X$ is
a convex body in $\mathbb{R}_{x}^{n}$ then its polar dual \textit{\ }%
$X^{\hbar}$ is \textit{the set of all }$p\in\mathbb{R}_{p}^{n}$ \textit{ }such
that $p\cdot x\leq\hbar$ for all $x\in X$.).

\section{The Donoho--Stark Uncertainty Principle}

\subsection{Statement}

Let $\psi\in L^{2}(\mathbb{R}^{n})$ and $\widehat{\psi}$ its unitary $\hbar
$-Fourier transform:
\[
\widehat{\psi}(p)=\left(  \tfrac{1}{2\pi\hbar}\right)  ^{n/2}\int%
_{\mathbb{R}^{n}}e^{-ipx/\hbar}\psi(x)dx.
\]
Let $\varepsilon\in\lbrack0,1]$. We say that $\psi$ is $\varepsilon
$-concentrated in a measurable set $X\subset\mathbb{R}^{n}$ if we have%
\[
\left(  \int_{\overline{X}}|\psi(x)|^{2}dx\right)  ^{1/2}\leq\varepsilon
||\psi||_{L^{2}}%
\]
where $\overline{X}=\mathbb{R}^{n}\diagdown X$ is the complement of $X$.
Similarly, $\widehat{\psi}$ is $\eta$-concentrated in $P$ if
\[
\left(  \int_{\overline{P}}|\widehat{\psi}(p)|^{2}dp\right)  ^{1/2}\leq
\eta||\widehat{\psi}||_{L^{2}}.
\]
These both inequalities are trivially equivalent to
\[
\int_{X}|\psi(x)|^{2}dx\geq(1-\varepsilon^{2})||\psi||_{L^{2}}^{2}\text{
},\text{ }\int_{P}|\widehat{\psi}(p)|^{2}dp\geq(1-\eta^{2})||\widehat{\psi
}||_{L^{2}}^{2}.
\]

In \cite{DS} Donoho and Stark proved the following result about the
concentration of a square integrable function and its Fourier transform:

\begin{theorem}
[Donoh--Stark]If $X\subset\mathbb{R}^{n}$ and $P\subset(\mathbb{R}^{n})^{\ast
}$ are measurable sets such that
\[
\int_{X}|\psi(x)|^{2}dx\geq(1-\varepsilon^{2})||\psi||_{L^{2}}^{2}\text{
},\text{ }\int_{P}|\widehat{\psi}(p)|^{2}dp\geq(1-\eta^{2})||\widehat{\psi
}||_{L^{2}}^{2}%
\]
where $\varepsilon$, $\eta\geq0$, then
\[
(\operatorname*{Vol}X)(\operatorname*{Vol}P)\geq(2\pi\hbar)^{n}(1-\varepsilon
-\eta)^{2}.
\]

\end{theorem}

A concise and limpid proof is given in Gr\"{o}chenig's treatise; \cite{gro};
in \cite{Boggiatto} Boggiatto \textit{et al}. somewhat extend and refine this result.

\subsection{Application to the Wigner transform}

For $\psi\in L^{2}(\mathbb{R}^{n})$ the Wigner function of $\psi$ is is
defined by the absolutely convergent integral%
\[
W\psi(z)=\left(  \tfrac{1}{2\pi\hbar}\right)  ^{n}\int_{\mathbb{R}^{n}%
}e^{-\frac{i}{\hbar}p\cdot y}\psi(x+\tfrac{1}{2}y)\overline{\psi(x-\tfrac
{1}{2}y)}dy;
\]
here $z=(x,p)\in\mathbb{R}^{2n}$ is the phase space variable.

\begin{corollary}
Assume that $\psi\in L^{2}(\mathbb{R}^{n})$ is even: $\psi(-x)=\psi(x)$. (i)
Assume that $W\psi$ is $\varepsilon$-concentrated in the measurable set
$X\subset\mathbb{R}^{2n}$:
\[
\int_{X}|W\psi(z)|^{2}dz\geq(1-\varepsilon^{2})||W\psi||_{L^{2}(\mathbb{R}%
^{2n})}^{2}.
\]
Then
\[
\operatorname*{Vol}(X)\geq(\pi\hbar)^{n}|1-2\varepsilon|.
\]
(ii) In particular, if $X=B^{2n}(\sqrt{\hbar})$ then
\[
\frac{1}{2}-\delta(n)\leq\varepsilon\leq\frac{1}{2}+\delta(n)
\]
where $\delta(n)=1/n!$
\end{corollary}

\begin{proof}
The ambiguity function $A\psi$ is defined by
\[
A\psi(z)=\left(  \tfrac{1}{2\pi\hbar}\right)  ^{n}\int_{\mathbb{R}^{n}%
}e^{-\frac{i}{\hbar}p\cdot y}\psi(y+\tfrac{1}{2}x)\overline{\psi(y-\tfrac
{1}{2}x)}dy
\]
and we have $A\psi(z)=FW\psi(Jz)$. A simple calculation shows that when $\psi$
is even,
\[
FW\psi(z)=A\psi(-Jz)=2^{-n}W\psi(-\frac{1}{2}Jz).
\]
and hence, for a measurable set $X\subset\mathbb{R}^{2n}$, seeing $U=-\frac
{1}{2}Jz$
\begin{align*}
\int_{2JX}|FW\psi(z)|^{2}dz  &  =2^{-2n}\int_{2JX}|W\psi(-\frac{1}{2}%
Jz)|^{2}dz\\
&  =.\int_{X}|W\psi(u)|^{2}du\geq1-\varepsilon^{2}.
\end{align*}
It follows from the Donoho--Stark uncertainty principle tat%
\[
\operatorname*{Vol}X)(\operatorname*{Vol}(2X\geq(2\pi\hbar)^{2n}%
(1-2\varepsilon)^{2}%
\]
that is
\[
\operatorname*{Vol}(X)\geq(\pi\hbar)^{n}|1-2\varepsilon)|.
\]
Choosing $X=B^{2n}(\sqrt{\hbar})$ we have $\operatorname*{Vol}(X)=(\pi
\hbar)^{n}/n!$ and hence $\frac{1}{n!}\geq(|1-2\varepsilon)|$ that is
\[
\frac{1}{2}-\frac{1}{2n!}\leq\varepsilon\leq\frac{1}{2}+.\frac{1}{2n!}%
\]
which is the same thing as%
\[
\delta(n)\leq\varepsilon\leq\frac{1}{2}+\delta(n)
\]
where $\delta(n)=1/n!$.
\end{proof}

\section{Polar Duality and Applications}

\subsection{Definition and main properties}

Let $X\subset\mathbb{R}_{x}^{n}$ be symmetric convex boy; that is $X$ is
compact and convex with non empty interior, and $X=-X$ (hence $0\in X$).

The polar dual $X^{\hbar}$ of $X$ with respect to its center $0$ is the set of
all $p=(p_{1},...,p_{n})$ in momentum space $\mathbb{R}_{p}^{n}=(\mathbb{R}%
_{x}^{n})^{\ast}$ such that%
\[
px=p_{1}x_{1}+\cdot\cdot\cdot+p_{n}x_{n}\leq\hbar.
\]

We have $(X^{\hbar})^{\hbar}=X$,and $X\subset Y$ implies $Y\subset X^{\hbar}$.
If $A\in GL(n,\mathbb{R})$, then
\[
(AX)^{\hbar}=(A^{T})^{-1}X^{\hbar}%
\]
hence, if $A=A^{T}>0$%
\[
\{x:Ax\cdot x\leq\hbar\}^{\hbar}=\{p:A^{-1}p\cdot p\leq\hbar\}.
\]
In particular
\[
B^{n}(\sqrt{\hbar})^{\hbar}=B^{n}(\sqrt{\hbar})
\]
($B^{n}(\sqrt{\hbar})$ is the only fixed point for polar duality relation
$X\longrightarrow X^{\hbar}$).

See Vershynin \cite{Vershynin} for a detailed study of the notion of polar
duality in the context of geometric analysis. Also see \cite{gopol2}.

So far we have assumed that the convex body $X$ is symmetric and hence
centered at zero. The general case is more difficult to handle and needs the
use of the Santal\'{o} point \cite{Santalo} as center. Its definition goes as
follows: or an arbitrary point $x_{0}$ in the interior of $X$ we define the
\textit{polar body of} $X$ \textit{with respect to} $x_{0}$ as being the set%
\begin{equation}
X^{\hbar}(x_{0})=(X-x_{0})^{\hbar}.\label{Santalopt2}%
\end{equation}
Santal\'{o} proved in \cite{Santalo} the following remarkable result: there
exists a \emph{unique} interior point $x_{\mathrm{S}}$ of $X$ (the
\textquotedblleft Santal\'{o} point of $X$\textquotedblright) such that the
polar dual $X^{\hbar}(x_{\mathrm{S}})=(X-x_{\mathrm{S}})^{\hbar}$ has centroid
$\overline{p}=0$ and its volume $\operatorname*{Vol}\nolimits_{n}(X^{\hbar
}(x_{\mathrm{S}}))$ is minimal for all possible interior points $x_{0}$. See
our discussion in \ \cite{gopol2} for details.

Let $X:Ax\cdot x\leq\hbar$ and $X^{\hbar}:(A^{-1})^{T}x\cdot x\leq\hbar$ be
dual ellipsoids. Then the John ellipsoid of the convex set $X\times X^{\hbar}$
is a "quantum blob" i.e. the image of the phase space ball $B^{2}(\sqrt{\hbar
})$ by a linear symplectic transformation $S\in\operatorname*{Sp}(n).$

Moreover, the \emph{Gromov width} of $X\times X^{\hbar}$ is
\[
w(X\times X^{\hbar})=4\hbar
\]
(in the case $n=1$ the sets $X$ and $X^{\hbar}$ are intervals and the area of
the rectangle $\times$ is $4\hbar$.

The definition of polar duality extends to the case where $\mathbb{R}_{x}^{n}$
and $\mathbb{R}_{p}^{n}=(\mathbb{R}_{x}^{n})^{\ast}$ are replaced with a pair
$(\ell.\ell^{\prime})$ of transverse Lagrangian planes (i.e. $\dim\ell
=\dim\ell^{\prime}=n$ and the symplectic form $\sigma$ vanishes on $\ell$
(resp. $\ell^{\prime}$): For a symmetric convex body $X_{\ell}\subset\ell$ he
polar dual $(X_{\ell})_{\ell^{\prime}}^{\hbar}$ is defined by
\[
(X_{\ell})_{\ell^{\prime}}^{\hbar}=\left\{  z^{\prime}\in\ell^{\prime}%
:\sigma(z.z^{\prime})\leq\hbar\right\}  .
\]
Again the the John ellipsoid of product $X_{\ell}\times(X_{\ell}%
)_{\ell^{\prime}}^{\hbar}$ contains a quantum blob $S(B^{2}(\sqrt{\hbar}))$
when $X_{\ell}$ is an ellipsoid. For details see e.g \cite{gopol4}.

\subsection{Relation with the uncertainty principle}

Recall \cite{go09,blobs,golu09} that a \textquotedblleft quantum
blob\textquotedblright\ is the image of the phase space ball $B^{2n}%
(\sqrt{\hbar})$ by a linear symplectic transformation $S\in\operatorname*{Sp}%
(n)$.This is easily seen as follows: the ellipsoid $X:Ax\cdot x\leq\hbar$ is
the image of the ball $B_{X}^{n}(\sqrt{\hbar})$ by the linear mapping
\ $A^{-1/2}$ while the ellipsoid $X^{\hbar}:A^{-1}p\cdot p\leq\hbar$ is that
of $B_{P}^{n}(\sqrt{\hbar})$ by $A^{1/2}$. It follows that the cell $X\times$
$X^{\hbar}$ is the image of the product $B_{X}^{n}(\sqrt{\hbar})\times
B_{P}^{n}(\sqrt{\hbar})$ by the symplectic mapping
\[
S=%
\begin{pmatrix}
A^{-1/2} & 0\\
0 & A^{1/2}%
\end{pmatrix}
.
\]
Now the unique largest ellipsoid (The John ellipsoid, see \cite{Ball})
contained in the convex set $B_{X}^{n}(\sqrt{\hbar})\times B_{P}^{n}%
(\sqrt{\hbar})$ is $B^{2n}(\sqrt{\hbar})$ hence the John ellipsoid of
$X\times$ $X^{\hbar}$ is $S(B^{2n}(\sqrt{\hbar})),$ which is a quantum blob.

\ Quantum blobs represents the smallest unit of phase space compatible with
the uncertainty principle. They is defined in the context of the
Robertson-Schr\"{o}dinger uncertainty relation and are characterized by its
invariance under symplectic transformations. To every quantum blob one
associates in a canonical way a generalized coherent state. For instance , to
the ball $B^{2n}(\sqrt{\hbar})$ is associated the $n$-dimensional coherent
state $\phi_{0}(x)=(\pi\hbar)^{-n/4}e^{-x^{2}/2\hslash}$. We have detailed
these properties in our paper \cite{golu09} with Luef.

Another illustration of is given by Hardy's uncertainty principle. It says
that if $A,B\in GL(n,\mathbb{R})$ are symmetric and positive definite and if
$\psi,\widehat{\psi}(\in L^{2}(\mathbb{R}^{n})$ satisfies the estimates
\[
|\psi(x)|\leq ke^{-Ax\cdot x/2\hbar}\text{ \ and \ }|\widehat{\psi}(p)|\leq
ke^{-Bp\cdot p/2\hbar}%
\]
for some $k>0$, then the ellipsoids $X=\{x:Ax\cdot x\leq\hbar\}$ and
$P=\{p:Bp\cdot p\leq\hbar\}$ satisfy $X^{\hbar}\subset P$ with equality
$P=X^{\hbar}$ if and only if $\psi(x)=Ce^{-Ax\cdot x/2\hbar}$ for some $C>0$.
The conditions on $|\psi(x)|$ and $|\widehat{\psi}(p)|$ are equivalent to
saying that the eigenvalues of $AB$ are $\leq1$; this is in turn equivalent to
$X^{\hbar}\subset P$; see \cite{ACHA} for a detailed proof.

\subsection{The Mahler volume and the Blaschke--Santal\'{o} inequality}

A remarkable property of polar duality, the Blaschke--Santal\'{o} inequality:
says that if $X$ is a symmetric convex body; then the Mahler volume $v(X)$,
defined by%
\[
v(X)=(\operatorname*{Vol}X)(\operatorname*{Vol}X^{\hbar})
\]
satisfies the inequality%
\begin{equation}
v(X)\leq(\operatorname*{Vol}\nolimits_{n}(B^{n}(\sqrt{\hbar}))^{2}
\label{santalo1}%
\end{equation}
that is,
\begin{equation}
v(X)=(\operatorname*{Vol}X)(\operatorname*{Vol}X^{\hbar})\leq\frac{(\pi
\hbar)^{n}}{\Gamma(\frac{n}{2}+1)^{2}} \label{BS0}%
\end{equation}
where $\operatorname*{Vol}\nolimits_{n}$ is the standard Lebesgue measure on
$\mathbb{R}^{n}$, and equality is attained if and only if $X\subset
\mathbb{R}_{x}^{n}$ is an ellipsoid centered at the origin. The Mahler has
conjectured that lower bound is%
\begin{equation}
\upsilon(X)\geq\frac{(4\hbar)^{n}}{n!} \label{volvo3}%
\end{equation}
but this claim has so far resisted to all proof attempts The best know result
is the following, due to Kuperberg \cite{Kuper}, who has shown that
\begin{equation}
\upsilon(X)\geq\frac{(\pi\hbar)^{n}}{4^{n}n!}. \label{kuper}%
\end{equation}
Summarizing, we have the bounds%
\begin{equation}
\frac{(\pi\hbar)^{n}}{4^{n}n!}\leq\upsilon(X)\leq\frac{(\pi\hbar)^{n}}%
{\Gamma(\frac{n}{2}+1)^{2}} \label{bounds}%
\end{equation}
\ (see \cite{gopolar} for a discussion of other partial results).

The Mahler volume has the intuitive interpretation as being a measure of
\textquotedblleft roundness\textquotedblright: its largest value is taken by
balls (or ellipsoids), and its smallest value (the bound (\ref{volvo3})) is
indeed attained by any $n$-parallelepiped%
\begin{equation}
X=[-\sqrt{2\sigma_{x_{1}x_{1}}},\sqrt{2\sigma_{x_{1}x_{1}}}]\times\cdot
\cdot\cdot\times\lbrack-\sqrt{2\sigma_{x_{n}x_{n}}},\sqrt{2\sigma_{x_{n}x_{n}%
}}]. \label{interval}%
\end{equation}
This is related to the covariances of the tensor product $\psi=\phi_{1}%
\otimes\cdot\cdot\cdot\otimes\phi_{n}$ of standard one-dimensional Gaussians
$\phi_{j}(x)=(\pi\hbar)^{-1/4}e^{-x_{j}^{2}/2\hbar}$; the function $\psi$ is a
minimal uncertainty quantum state in the sense that it reduces the Heisenberg
inequalities to equalities. We suggest that such quantum considerations might
lead to proof of Mahler's conjecture.

\subsection{A concentration result}

Let us prove:

\begin{theorem}
Let $X$ be a symmetric body in $\mathbb{R}^{n}$ and $\psi\in L^{2}%
(\mathbb{R}^{n})$. The concentration inequalities%
\begin{equation}
\left(  \int_{X}|\psi(x)|^{2}dx\right)  ^{1/2}\leq\varepsilon||\psi||_{L^{2}%
}\text{ }\ \text{, \ }\int_{X^{\hbar}}|\widehat{\psi}(p)|^{2}dp\leq
\eta||\widehat{\psi}||_{L^{2}}%
\end{equation}
can hold if if and only if
\begin{equation}
1-\delta(n)\leq\varepsilon+\eta\leq1+\delta(n)
\end{equation}
where
\[
\delta(n)=\frac{1}{2^{n/2}\Gamma(n/2+1)}.\overset{n\rightarrow\infty
}{\rightarrow}0.
\]

\end{theorem}

\begin{proof}
Combining the Donoho--Stark and the Blaschke--Santal\'{o} inequalities yields%
\begin{equation}
(2\pi\hbar)^{n}(1-\varepsilon-\eta)^{2}\leq(\operatorname*{Vol}%
X)(\operatorname*{Vol}X^{\hbar})\leq\frac{(\pi\hbar)^{n}}{\Gamma(\frac{n}%
{2}+1)^{2}}\label{DSBS}%
\end{equation}
which implies that $\varepsilon+\eta\overset{n\rightarrow\infty}{\backsim}1$.
More precisely, setting
\[
\delta(n)=\frac{1}{2^{n/2}\Gamma(n/2+1)}%
\]
\ we have%
\begin{equation}
1-\delta(n)\leq\varepsilon+\eta\leq1+\delta(n).\label{epsilon}%
\end{equation}

\end{proof}

This result has a simple quantum-mechanical interpretation. Take
$||\psi||_{L^{2}}=1$ and define, as is usual in quantum mechanics, the pure
state presence probabilities%
\[
\Pr(x\in X)=\int_{X}|\psi(x)|^{2}dx\text{ \ },\text{ \ }\Pr(p\in X^{\hbar
})=\int_{X^{\hbar}}|\widehat{\psi}(p)|^{2}dp.
\]
Assume that $\Pr(x\in X)\geq1-\varepsilon^{2}$ and $\Pr(p\in X^{\hbar}%
)\geq1-\eta^{2}$. the result above implies that for large $n$ we have%
\[
\varepsilon+\eta\approx\frac{1}{2}.
\]

One can loosely say that the more the quantum state represented by $\psi$ is
localized in $X$ in position representation, the less it is localized in the
polar dual \ $X^{\hbar}$ in momentum representation.

\begin{acknowledgement}
This work has been financed by the Austrian Research Foundation FWF (Grant
number PAT 2056623). It was done during a stay of the author at the Acoustics
Research Institute group at the Austrian Academy of Sciences.
\end{acknowledgement}

\textbf{DATA\ AVAILABILITY\ STATEMENT}: no data has been used created, other
that the source file

\textbf{CONFLICT\ OF\ INTERESTS}; there are no conflict of interests

MauriceAlexis.deGossondeVarennes@oeaw.ac.at

maurice.de.gosson@univie.ac.at

\end{document}